\documentclass{amsart}
\usepackage[english]{babel}
\usepackage{bbm}
\usepackage{graphics}
\usepackage{hyperref}

\usepackage{tikz}
\usetikzlibrary{patterns}
\usepackage{mathtools}

\newcommand{\markedN}[1]{\ensuremath{\overline{\cal N}_{#1}}}

\newcommand\coloneq{\stackrel{\text{\rm def.}}{=}} 
\newcommand{\steq}[1]{\stackrel{\text{\rm #1.}}{=}} 

\newcommand{\N}{\ensuremath{\mathbb{N}}}
\newcommand{\R}{\ensuremath{\mathbb{R}}}

\renewcommand{\P}{\ensuremath{\mathbb{P}}}

\newcommand{\ind}[1]{\ensuremath{\mathbbm{1}_{\{#1\}}}}
\newcommand{\diff}{\mathop{}\mathopen{}\mathrm{d}}
\newcommand{\cal}[1]{\ensuremath{\mathcal{#1}}}

\def\eps{\varepsilon}

\newtheorem{proposition}{Proposition}

\newtheorem{corollary}{Corollary}

\newtheorem{theorem}{Theorem}

\title{A Stochastic Analysis of a Network with Two Levels of Service}

\author{Vianney B{\oe}uf}
\address[V. B{\oe}uf]{\'Ecole des Ponts ParisTech, INRIA and CMAP, \'Ecole Polytechnique,
CNRS, Brigade de Sapeurs-Pompiers de Paris. CMAP, \'Ecole Polytechnique, 91128 Palaiseau Cedex, France}
\email{Vianney.Boeuf@inria.fr}
\urladdr{http://www.cmap.polytechnique.fr/~boeuf}
\author{Philippe Robert}
\address[Ph. Robert]{INRIA, 2 rue Simone Iff, CS 42112, 75589 Paris Cedex 12, France}
\email{Philippe.Robert@inria.fr}
\urladdr{http://team.inria.fr/rap/robert}

\date{\today}
\begin{document}

\begin{abstract}
In this paper a stochastic model of a call center with a two-level architecture is analyzed. A first-level pool of operators answers calls, identifies, and handles non-urgent calls. A call classified as urgent has to be transferred  to specialized operators at the second level. When the operators of the second level are all busy, the operator of first level handling the urgent call is blocked  until  an operator at the second level is available. Under a scaling assumption, the evolution of the number of urgent calls blocked at level~$1$ is investigated. It is shown that if the ratio of the number of operators at level $2$  and~$1$ is greater than some threshold, then, essentially, the system operates without congestion, with probability close to $1$, no urgent call is blocked after some finite time.  Otherwise, we prove that a positive fraction of the operators of the first level are blocked due to the congestion of the second level. Stochastic calculus with Poisson processes, coupling arguments and formulations in terms of Skorokhod problems are the main mathematical tools  to establish these convergence results.

\end{abstract}

\maketitle

\hrule

\vspace{-3mm}

\tableofcontents

\vspace{-10mm}

\hrule

\bigskip

\section{Introduction}
The motivation of the model analyzed in this paper originates from a collaboration with  ``Pr\'efecture de police de Paris'', the police department of Paris, and ``Brigade de sapeurs-pompiers de Paris'', the fire department of Paris, to design  an emergency call center in charge of receiving emergency calls for police {\em and} for firemen in Paris area.  The previous organization had two independent call centers with a single level of operators. The new call center  has an architecture with  two levels of operators.  A first-level pool of operators  handles (numerous)  non-urgent calls and has to detect and  transfer calls classified as  urgent to a  second-level pool of more specialized operators,  policemen or firemen, depending on the nature of the call.  Second level operators may dispatch emergency means, if needed.  The first level pool operates therefore as a filter so that the second-level pool can process efficiently urgent calls. An additional, natural, constraint is that if a first level operator has detected an urgent call,  this  operator releases the call only  when  a second level operator has handled it. In particular, the operator will wait when all servers of the second level are busy. In this situation there are two issues: firstly, the handling of the urgent call is delayed and, secondly, the server of the first level is blocked  and, consequently, the processing capacity of the first level is reduced. The main problem in the design of this new organization is of determining a minimal number of (expensive) second level operators  necessary so that this blocking phenomenon  has a small probability.

We will investigate the behavior of this architecture in stressed situations, i.e.,  when a large number of incoming calls is arriving at the first level.  A key characteristic to analyze in this situation is the evolution of the number of blocked operators at level~$1$. This number should remain small in a convenient design. For this  reason, it will be assumed that an infinite number of calls are waiting for processing in a queue.  Calls require random processing time whose distribution depends on the level and the class of the call (urgent or non-urgent). We now give a quick description of this system in terms of a queueing model.

\subsection*{A Queueing Description of the System}
As input, there is an infinite queue of jobs waiting to enter the system, this is the saturation assumption mentioned above. With probability $p{\in}[0,1]$ a job is of class~$0$, otherwise it is of class~$1$. 
A job of class~$0$ represents an urgent call, otherwise it is a non-urgent call.
\begin{enumerate}
\item The first level has $C_1$ servers. \\
  Every time a server of this level is idle, it immediately receives a job from the infinite queue. It is of class $i{\in}\{0{,}1\}$ with probability $p{\in}[0{,}1]$ and $1{-}p$, respectively. A job of class $i$ requires an exponentially distributed service with rate $\mu_{i1}$ at this level.

  \noindent
 Class~$0$ jobs are urgent calls and have to be processed by level~$2$.
 \begin{enumerate}
 \item  When a job of class $0$ completes its service at level~$1$, it goes to the second level if there is at least one idle server there.
   \item If there is no place then it remains at the first level and, consequently, blocks a server at this level.  As soon as a job leaves the second level, a blocked job at the first level is sent to the second level and the server can take a new job in the infinite queue.
    \end{enumerate}
  When a job of class $1$ completes its service, it leaves the system.
\item The second level has $C_2$ servers and receives only class~$0$ jobs.  A job at this level requires a exponentially distributed service with rate $\mu_{02}$.
\end{enumerate}
See Figure~\ref{Fig1}.
A key feature of this network is that blocked jobs of class $0$ at level $1$ reduce the capacity of the system since the corresponding servers at level $1$ cannot process the calls waiting in the saturated queue.

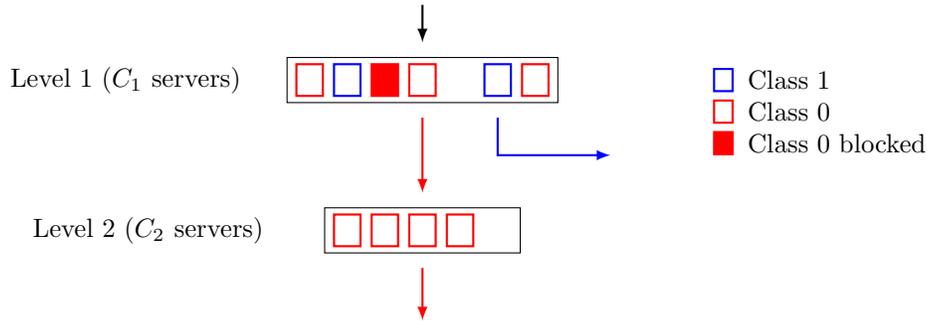
\begin{figure}[ht]
\begin{tikzpicture}
\tikzset{job/.style={rectangle,thick,minimum width=10pt,minimum height=12pt}}
\begin{scope}
    \node[anchor=east] at (-2.3,0) {Level 1 ($C_1$ servers)};
    \draw (-1.8,-.3) rectangle (1.8,.3);
    \node [job,draw=red] at (-1.5,0) {};

    \node [job,draw=blue] at (-1,0) {};
    \node [job,draw=red,fill=red] at (-0.5,0) {};
    \node [job,draw=red] at (0,0) {};
    \node [job,draw=blue] at (1,0) {};
    \node [job,draw=red] at (1.5,0) {};
\end{scope}

\begin{scope}[shift={(0,-2)}]
    \node[anchor=east] at (-2,0) {Level 2 ($C_2$ servers)};
    \draw (-1.3,-.3) rectangle (1.3,.3);
    \node [job,draw=red] at (-1,0) {};
    \node[job,draw=red] at (-0.5,0) {};
    \node[job,draw=red] at (0,0) {};
    \node [job,draw=red] at (.5,0) {};
\end{scope}

\draw[->,>=latex,red,thick] (0,- 0.5) -- (0,-1.5);
\draw[<-,>=latex,thick] (0,.5) -- (0,1);
\draw[->,>=latex,blue,thick] (1,-0.5) -- (1,-1) -- (2.5,-1);
\draw[->,>=latex,red,thick] (0,-2.5) -- (0,-3.2);

\begin{scope}[shift={(4,0)},scale=.7]
    \node[job,draw=blue,scale=0.7] (leg1) at (0,0) {};
    \node[job,draw=red,scale=0.7] (leg2) at (0,-.6) {};
    \node[job,draw=red,fill=red,scale=0.7] (leg3) at (0,-1.2) {};
    \node[anchor=west,shift={(0.2,0)}] at (leg1) {Class 1};
    \node[anchor=west,shift={(0.2,0)}] at (leg2) {Class 0};
    \node[anchor=west,shift={(0.2,0)}] at (leg3) {Class 0 blocked};
\end{scope}
\end{tikzpicture}
\caption{Queueing System with Two Levels}\label{Fig1}
\end{figure}

\subsection*{Literature}
\subsubsection*{Deterministic Modeling}
In this paper, the classes of calls and their processing times are assumed to be random. In a non-random setting, some aspects of this system have been investigated in Allamigeon et al.~\cite{Abg2015,Abg2016} where a performance analysis was carried out using a deterministic Petri net modeling.  A Petri net is a language describing systems in which resources circulate from place to place, incurring concurrency, synchronizations and bifurcations~\cite{murata1989petri,best2013structure}.  The dynamics of a Petri net can be translated into a dynamical system, whose stability and stable points can be analyzed, see Cohen et al.~\cite{cohen1995asymptotic}. In Allamigeon et al.~\cite{Abg2015,Abg2016}, a simplified model of the emergency call center is investigated.  Computations on the stationary regimes of the dynamical system have shown a phase transition characterizing the different levels of congestion of the call center.  The threshold is a critical ratio between the number of operators at level~$1$ and level~$2$.

The analysis of Petri net models may give general results for this class of systems in a deterministic framework.  However, the dynamics investigated in the above articles do not take account of the random nature  of the delays or the  classes of calls in the call center for example.  In contrast, the queueing network analysis adopted in the present article focuses on a simpler system describing the transfer or the blocking of calls from level~1 to level~2. As it will be seen, it provides a deeper understanding of the behavior of this system in a random context.

It should be noted that the dynamical system obtained below in our results as the scaling limit of our model do correspond to the differential equations and the multi-phase behavior resulting from the continuous Petri net modeling of Allamigeon et al.~\cite{Abg2016}, hence highlighting the consistency and strong relationship between both analyses.

\subsubsection*{Queueing Models of Blocking Phenomena}
A natural class of stochastic models related to the system described above is that of call centers. There is a huge literature dealing with the problem of staffing  these systems. To the best of our knowledge, few seem to have considered jobs  going through a series of call centers as in our case. The closest models of this literature seem to be multi-skill call centers where jobs can have different levels of quality of service depending on the call center chosen. They are nevertheless addressing quite different problems than the ones considered in this paper. See Koole and Mandelbaum~\cite{Koole} for a survey.

The model that we are studying can be described in terms of finite capacity queues with blocking in tandem. The blocking has the effect that, when a server at level~$1$ completes the service of a class~$0$ job, it cannot be used again until a server at level~$2$ is available. At level~$1$ a fraction of the servers, and consequently the corresponding calls, may be blocked. Related models have been investigated in the literature, see the survey Balsamo~\cite{Balsamo}.  The papers study  the corresponding finite Markovian models of these systems  to express in particular the blocking  probability at equilibrium.  The corresponding equilibrium equations  do not have, in general,  a solution with a closed form expression.  When the values of the capacities (the numbers of operators) are not small, the dimension of the state space can be quite large so that a numerical procedure can also be out of reach in practice. Some approximations  have been proposed but, for the moment, without any convergence result which could give an idea of the accuracy of such estimations. Kelly~\cite{Kelly} has investigated the problem of blocking of a series of queues, the analysis is concentrated on the estimation, via bounds, of the achievable throughput of such a system. To conclude, the literature of rigorous mathematical results for  finite capacity queues with blocking is therefore somewhat scarce.

When the blocking is replaced by the following mechanism defined as an exclusion process:  a job  blocked at some stage immediately repeats a service until the next stage can accommodate it. Some of the mathematical models related to the asymmetric simple exclusion process can give some insights on the performances of these systems. Due to its relative mathematical tractability, the literature investigating these processes is also huge. See, for example Liggett~\cite{BLiggett} for a general presentation of these important processes and Liggett~\cite{Liggett} for a study of  asymmetric simple exclusion process in finite dimension. These models are however quite different and do not seem to be usable since blocking phenomenon of interest is not really taken into account.

\subsection*{Contributions}
With the above notations for our system, one of the main results of the paper,   Theorem~\ref{theostab},  shows that, under appropriate scaling conditions, if  $r $ is the ratio of the capacities of the two levels, $r{=}{C_2}/{C_1}$,   then  the  condition
\begin{equation}\label{eq:cond1}
r\left(\frac{p}{\mu_{01}}+\frac{1{-}p}{\mu_{11}}\right)>\frac{p}{\mu_{02}}
\end{equation}
implies that there exists some fixed instant independent of the initial state such that after that time, with high probability, there are no blocked customers at level~$1$ on any finite time interval. See Theorem~\ref{theostab} and Corollary~\ref{thm:y-star-is-null}.

Otherwise, if the opposite (strict) inequality
\begin{equation}\label{eq:cond2}
r\left(\frac{p}{\mu_{01}}+\frac{1{-}p}{\mu_{11}}\right)<\frac{p}{\mu_{02}}
\end{equation}
holds then, Theorem~\ref{theoinstab}  shows that,  under appropriate scaling conditions,  the fraction of blocked customers at level~$1$ is positive after some time almost surely and it converges to
\[
 1 -  \dfrac{\mu_{02}}{\mu_{01}}\frac{C_2}{C_1} \left( \dfrac{(1-p)\mu_{01}}{p \mu_{11}} + 1 \right).
 \]
 See also Corollary~\ref{thm:z-is-null}.
 
Consequently, as the intuition suggests, if the ratio $C_2/C_1$ of the number  of servers  is larger than some parameter then the phenomenon of blocking will not occur with high probability.  Relation~\eqref{eq:cond1} gives  therefore a rule for a convenient design of such a system.

\subsubsection*{A Heuristic Picture} Assume that there is no blocking at level~$1$ of class~$0$ jobs. Level~$1$ can be seen as a simple birth and death process described by the number of jobs $(Q(t))$ of class $0$. A birth (resp.\ death) occurs when a job of class~$1$ (resp.\ $0$)  completing its service is replaced by a job of class~$0$ (resp.\ $1$). Therefore in state $x{\in}\{0,\ldots,C_1\}$, the birth rate is  $p(C_1-x)\mu_{11}$ and the death rate is $(1-p)\mu_{01}x$.
At equilibrium these two rates should be of the same order and therefore that, for a large $C_1$, the number $Q_0$ of class~$0$ jobs is of the order of
\[
Q_0\sim C_1\frac{p \mu_{11}}{(1-p)\mu_{01}+p\mu_{11}}.
\]
To avoid congestion, the rate $\mu_{01}Q_0$ at which class~$0$ jobs enter level $2$ must be smaller than the maximal output rate of the second level, that is $C_2\mu_{02}$. This gives exactly Condition~\eqref{eq:cond1}.

\subsubsection*{Mathematical Aspects}
Proving rigorously  these intuitive results turns out to be, quite surprisingly,  challenging.  The Markov process associated with the queueing system has a finite state space of $\N^3$. It bears some similarity with classical loss networks of the literature  but with a routing mechanism as in Jackson networks. See Kelly~\cite{BKelly}.  As such, little can be said for this process, in particular its invariant probability distribution does not  seem to have a simple closed form expression. 

To get quantitative results on this system a scaling approach is used.  It is assumed that the capacities $C_1$ and $C_2$ are both large so that $C_2/C_1$ is close to some fixed constant $r{>}0$. In this framework one investigates convergence results concerning stochastic processes indexed by the scaling parameter $C_1$.  The main technical difficulties lie in the behavior of the processes at the boundaries of the state space, when there are no blocked customers at level~$1$ or when there are no idle servers at level~$2$.  As always with processes behaving locally as random walks, getting convergence results of scaled process in this context with two boundaries may be difficult. This situation has some similarities with the reflected random walks associated with classical queueing networks where the convergence results can be, sometimes, obtained by using a Skorokhod problem formulation. See Harrison and Reiman~\cite{Harrison}, Chen and Mandelbaum~\cite{Chen} or Section~9.4 of Robert~\cite{Robert} for example. There is no such global Skorokhod problem formulation for our model. An additional difficulty is the dependence on the scaling parameter of the location of time intervals where blocking (or no-blocking) occurs eventually.

To handle this complicated setting, we introduce two auxiliary processes which are first separately investigated in Section~\ref{sec:auxiliary-first}, for each of them, only one of the boundary conditions is involved. A generalized Skorokhod problem formulation is used in both cases. The final Section~\ref{sec:first} establishes the main convergence results.   Stochastic calculus with Poisson processes, coupling arguments and the results obtained on auxiliary processes are the main ingredients of the proofs. See the proof of Proposition~\ref{thm:yplusystar} for example. 
\section{The Stochastic Model}\label{sec:init}
To analyze the stability properties of this network, it will be assumed that the capacities of the two levels of service are large, proportional to a scaling parameter $N$.  Qualitative and quantitative properties of the system when $N$ gets large will be obtained. In particular we will determine the conditions on the parameters for which the blocking probability is negligible or not. We begin with a brief reminder on Poisson processes and some notations used in this domain.

\subsection*{Notations for Poisson processes}
Throughout the paper, for $\xi{>}0$, one denotes by  ${\cal N}_{\xi}{=}(t_n)$ a Poisson point  process on $\R_+$ with rate $\xi$ and $({\cal N}_{\xi,i})$ denotes a sequence of i.i.d.\ such Poisson processes. In the following, we will use at some occasions the following coupling of Poisson processes, for $0{<}\alpha{\leq}\beta$, one can construct a version of ${\cal N}_{\alpha}$ and ${\cal N}_{\beta}$ such that, for all $0{\leq}s{\leq}t$,
\[
{\cal N}_{\alpha}([s,t])\steq{def}\int_s^t {\cal N}_{\alpha}(\diff s)\leq {\cal N}_{\beta}([s,t]).
\]
This can be done in the following way. If ${\cal P}$ is a Poisson process on $\R_2^+$ whose intensity measure is Lebesgue on this space, then for $\xi{\in}\{\alpha,\beta\}$, the order relation will hold if we take
\[
{\cal N}_{\xi}(\diff t)={\cal P}([0,\xi]\times \diff t).
\]

The notation $\markedN{\xi}{=}(t_n,B_n)$ is for a marked point Poisson process on $\R_+{\times}\{0{,}1\}$, where $(t_n)$ is a Poisson process with rate $\xi$ on $\R_+$ and  $(B_n)$ is an i.i.d.\ sequence of Bernoulli random variables with parameter $p$.   If $f$ is some positive Borelian function on $\R_+{\times}\{0{,}1\}$, we will use the (usual) notation
\[
\int f(t,b),\markedN{\xi}(\diff s,\diff b)=\sum_{n\geq 1} f(t_n,B_n),
\]
$(\markedN{\xi,i})$ denotes a sequence  of such i.i.d.\ marked point Poisson processes.  Concerning marked point Poisson processes see Kingman~\cite{Kingman} for example. They can be interpreted as follows in our case, if $\xi{\in}\{\mu_{01}{,}\mu_{11}{,}\mu_{02}\}$,  $u{\in}\{0{,}1\}$ and the quantity $\markedN{\xi,i}(\diff t{,}\{u\})$ is not $0$, then a completion of a service occurs at time $t$, and if a new job enters the first level at this occasion, $u$ is the class of this job. Clearly the point process ${\cal N}_{\xi}(\diff t)$ has the same distribution as $\markedN{\xi}(\diff t{,}\{0{,}1\})$.

\subsection*{Scaling}
The capacities $C_1$ and $C_2$ of levels $1$ and $2$ depend on a scaling parameter $N$, $C_1{=}C_1^N{=}N$  and $C_2{=}C_2^N$ such that the convergence
\begin{equation}\label{Scaling}
\lim_{N\to+\infty}{C_2^N}/{C_1^N}=r
\end{equation}
holds for some $r{>}0$.

The evolution of the state of this system can described by the stochastic process $(X^N(t))\coloneq(Y_*^N(t),Y^{N}(t),Z^N(t))$ with, for $t{\geq}0$,
\begin{itemize}
  \item $(Y_*^N(t))$ being the number of class~$0$ jobs blocked at level~$1$ at time $t$,
\item $(Y^N(t))$, the number of class~$0$ jobs being served at level~$1$,
\item $(Z^N(t))$,  the number of idle servers at level~$2$.
\end{itemize}
For $t{\geq}0$, remark that at least one of the variables $Y_*^N(t)$ or $Z^N(t)$ is null.
It is not difficult to see that $(X^N(t))$ is  an irreducible Markov process on the state space
\[
{\cal S}_N\coloneq \left\{x=(y_*,y,z)\in\N^3: y+y_*\leq C_1^N, z\leq C_2^N, y_*\!\!\cdot z=0\right\}
\]
It will be assumed that the sequence of initial states satisfies the relation
\begin{equation}\label{ScalingInit}
\lim_{N\to+\infty}\frac{1}{N}(Y_*^N(0),Y^{N}(0),Z^N(0)){=}x_0{=}(y_{*0},y_0,z_0){\in}[0,1]^2{\times}[0,r].
\end{equation}
The vector $x_0$ will be referred to as the initial fluid  state in the following. 
The transition rates are defined as follows, for $x{=}(y_*,y,z){\in}{\cal S}$,
\begin{equation} \label{eq:Q_rates}
x\mapsto
\begin{cases}
(y_*{+}1,y{-}1,0) &\text{at rate } \quad \mu_{01}y \ind{z=0},\\
(0,y{-}1,z{-}1) &\text{"\phantom{  rate } } \quad \mu_{01}y(1{-}p) \ind{z>0},\\
(0,y,z{-}1) &\text{"\phantom{  rate } } \quad \mu_{01}yp\ind{z>0},\\
(y_*,y{+}1,z) &\text{"\phantom{  rate } } \quad \mu_{11} p(N{-}y_*{-}y),\\
(y_*{-}1,y,z) &\text{"\phantom{  rate } } \quad  \displaystyle (1{-}p)\mu_{02}C_2^N\ind{y_*>0},\\
(y_*{-}1,y{+}1,z) &\text{"\phantom{  rate } } \quad  \displaystyle p\mu_{02}C_2^N\ind{y_*>0},\\
(0,y,z{+}1) &\text{"\phantom{  rate } } \quad  \mu_{02}(C_2^N{-}z)\ind{y_*=0}.
\end{cases}
\end{equation}
Due to the constraints on the coordinates $y_*$ and $z$ of $x$ (at least one of them is $0$), the Markov process $(X^N(t))$ can be seen as a two-dimensional process as depicted in Figure~\ref{fig:domain}.

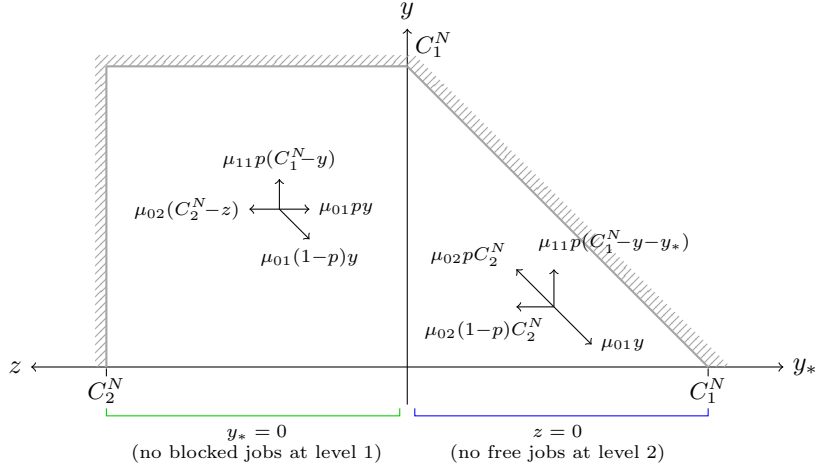
\begin{figure}
\begin{center}
\begin{tikzpicture}[scale=0.5]
    \node (p0) at (0,0) {};
    \draw[->] (p0.center) -- +(10,0) node[anchor=west] {$y_{*}$};
    \draw[->] (0,-1) -- +(0,10) node[anchor=south] {$y$};
    \draw[->] (p0.center) -- +(-10,0) node[anchor=east] {$z$};

    \node (p01) at (-1,0) {};
    \node (p10) at (1,0) {};
    \node (p08) at (-8,0) {};
    \node (p80) at (8,0) {};
    \node (p8) at (0,8) {};

    \draw (p08) -- +(0,-.04) node[anchor=north,font=\small] {$C_2^N$};
    \draw (p80) -- +(0,-.04) node[anchor=north,font=\small] {$C_1^N$};
    \draw (p8) -- +(-.04,0) node[anchor=south west,font=\small] {$C_1^N$};

    \draw[thick,gray!70!white] (p08.center) -- +(0,8);
    \draw[thick,gray!70!white] (p8.center) -- +(-8,0);
    \draw[thick,gray!70!white] (p80.center) -- (p8.center);
    \fill [pattern=north east lines,pattern color=gray!70!white] (-8.3,0) -- (p08.center) -- (-8,8) -- (0,8) -- (0,8.3) -- (-8.3,8.3) -- cycle;
    \fill [pattern=north east lines,pattern color=gray!70!white] (0,8) -- (0,8.3) -- (0.2,8.3) -- (8.6,0) -- (8,0) -- cycle;

    \draw[blue] (.2,-1.1) -- (.2,-1.3) -- (8,-1.3) -- (8,-1.1);
    \node at (-4,-1.3) [anchor=north,font=\scriptsize,align=center] {$y_*=0$\\(no blocked jobs at level~1)};
    \draw[green!80!black] (-.2,-1.1) -- (-.2,-1.3) -- (-8,-1.3) -- (-8,-1.1);
    \node at (4,-1.3) [anchor=north,font=\scriptsize,align=center] {$z=0$\\(no free jobs at level 2)};

    \node (d_or) at (-3.4,4.2) {};
    \draw[->] (d_or.center) -- +(-.8,0) node[anchor=east,font=\scriptsize] {$\mu_{02} (C_2^N\!{-}z)$};
    \draw[->] (d_or.center) -- +(0,.8) node[anchor=south,font=\scriptsize] {$\mu_{11} p (C_1^N\!\!{-}y)$};
    \draw[->] (d_or.center) -- +(+.8,0) node[anchor=west,font=\scriptsize] {$\mu_{01} p y$};
    \draw[->] (d_or.center) -- +(+.8,-.8) node[anchor=north,font=\scriptsize] {$\mu_{01} (1{-}p) y$};

    \node (g_or) at (3.9,1.6) {};
    \draw[->] (g_or.center) -- +(-1,1) node[anchor=east,font=\scriptsize,shift={(0,.2)}] {$\mu_{02} p C_2^N$};
    \draw[->] (g_or.center) -- +(-1,0) node[anchor=east,font=\scriptsize,shift={(0.5,-0.3)}] {$\mu_{02} (1{-}p) C_2^N$};
    \draw[->] (g_or.center) -- +(0,1) node[anchor=south,font=\scriptsize,shift={(.8,.1)}] {$\mu_{11} p (C_1^N\!\!{-}y{-}y_*)$}; %
    \draw[->] (g_or.center) -- +(1,-1) node[anchor=west,font=\scriptsize] {$\mu_{01} y$};
\end{tikzpicture}
\end{center}
\caption{A Representation of the Transitions Rates of $(X^N(t))$.
The three-dimensional  structure  is ``unfolded'' in two dimensions.} \label{fig:domain}
\end{figure}

\subsection*{Representation by Stochastic Differential Equations}
From the transition rates~\eqref{eq:Q_rates}, the process $(X^N(t))$ can also be seen as the unique solution of the following stochastic differential equations,
\begin{multline} \label{eq:integr-ystar-diff}
\diff Y_*^N(t) =  \sum_{i=1}^{+\infty} \ind{i\leq Y^N(t-),Z^N(t-)=0}\markedN{\mu_{01},i}(\diff t,\{0,1\})\\
- \ind{Y_*^N(t-)>0}\sum_{i=1}^{C_2^N}\markedN{\mu_{02},i}(\diff t,\{0,1\}),
\end{multline}
\vspace{-5mm}
\begin{multline}\label{eq:integr-y-diff}
  \diff Y^{N}(t) =-\sum_{i=1}^{+\infty}\ind{i\leq Y^N(t-)}\ind{Z^N(t-)=0}\markedN{\mu_{01},i}(\diff t,\{0,1\})
    \\ -\ind{Z^N(t-)>0}\markedN{\mu_{01},i}(\diff t,\{1\}) + \ind{Y_*^N(t-)>0}\sum_{i=1}^{C_2^N}\markedN{\mu_{02},i}(\diff t,\{0\})\\
  +\sum_{i=1}^{+\infty}\ind{i\leq N-Y_*^N(t-)-Y^{N}(t-)}\markedN{\mu_{11},i}(\diff t,\{0\}),
\end{multline}
\vspace{-5mm}
\begin{multline}\label{eq:integr-z-diff}
  \diff Z^{N}(t) =- \sum_{i=1}^{+\infty}\ind{i\leq Y^N(t-),Z^N(t-)>0} \markedN{\mu_{01},i}(\diff t,\{0,1\})\\
  +  \sum_{i=1}^{+\infty} \ind{i\leq C_2^N-Z^N(t-),Y_*^N(t-) = 0}\markedN{\mu_{02},i}(\diff t,\{0,1\}),
\end{multline}
starting from some fixed initial state. The notation $f(t{-})$ stands for the left-limit of $f$ at $t$.

\subsection*{Filtration}
The  $\sigma$-field ${\cal F}_t$ of the events up to time $t$ is classically  defined as the $\sigma$-field generated by the random variables
\[
\markedN{\xi,i}([0,s]{\times}u),\text{ where } \xi\in\{\mu_{01},\mu_{11},\mu_{02}\}, s\in[0,t], u\in\{\{0\},\{1\}\} \text{ and } i\in\N.
\]
With this definition the process $(Y^N_*(t),Y^N(t),Z^N(t))$ is clearly $({\cal F}_t)$-adapted. The martingale properties mentioned in the following are understood to be  with respect to this filtration.

\subsection*{Evolution equations}
The rescaled process is denoted by
\begin{equation}\label{eqX}
\left(\overline{X}^N(t)\right)\steq{def}\left(\overline{Y}_*^N(t),\overline{Y}^{N}(t), \overline{Z}^N(t)\right)\coloneq
\frac{1}{N}\left(Y_*^N(t),Y^{N}(t),Z^{N}(t)\right),
\end{equation}
the integration of the above SDEs and classical stochastic calculus give the relations
\begin{multline}
  \overline{Y}_*^N(t) =\overline{Y}_*^N(0) + \mu_{01}\int_0^t \overline{Y}^N(s)\ind{\overline{Z}^N(s)=0} \diff s
\\  -\mu_{02}\dfrac{C_2^N}{N}\int_0^t \ind{\overline{Y}_*^N(s)>0} \diff s + M_{Y_*}^{N}(t), \label{eq:integr-ystar}
\end{multline}
\vspace{-5mm}
\begin{multline}
    \overline{Y}^{N}(t) =\overline{Y}^{N}(0)-\mu_{01}\int_0^t \overline{Y}^N(s)\left(1{-}p\ind{\overline{Z}(s)>0}\right)\diff s
        \\ + p\mu_{02}\dfrac{C_2^N}{N}\int_0^t\ind{\overline{Y}_*^N(s)>0}\diff s
    +p\mu_{11}\int_0^t(1{-}\overline{Y}_*^N(s){-}\overline{Y}^{N}(s))\diff s + M_Y^N(t), \label{eq:integr-y}
\end{multline}
\vspace{-5mm}
\begin{multline}
  \overline{Z}^{N}(t) =\overline{Z}^{N}(0)
    -\mu_{01}\int_0^t \overline{Y}^N(s)\ind{\overline{Z}^N(s)>0}\diff s\\
    +\mu_{02}\int_0^t\left(\dfrac{C_2^N}{N}{-}\overline{Z}^N(s)\right)\ind{\overline{Y}_*^N(s)=0}\diff s + M^N_Z(t), \label{eq:integr-z}
\end{multline}
where, for $V{\in}\{Y_*,Y,Z\}$, $(M^N_V(t))$ is a martingale.
We complete this section with a tightness result.
\begin{proposition}\label{proptight}
  The sequence of processes $(\overline{X}^N(t))$ defined by Relation~\eqref{eqX} is tight and any of its limiting points is  a continuous process.
\end{proposition}
\begin{proof}
Since, for $t{\geq}0$,  one has ${Y}_*^{N}\!(t){+}{Y}^{N}\!(t){\leq} N$ and $Z^N\!(t){\leq} C_2^N$, the variables $\overline{Y}_*^N\!(t)$, $\overline{Y}^{N}\!(t)$ and $\overline{Z}^{N}\!(t)$ are  thus uniformly bounded. By using a similar procedure as in the proof of Theorem~6.13 page~159 of Robert~\cite{Robert}, one can show that the  expected value of the previsible increasing process of  the martingales  $(M^N_V(t))$,  $V{\in}\{Y_*,Y,Z\}$,  is of the order of $1/N$ and thus converges to $0$.
By Doob's Inequality, one gets that for any $\eta{>}0$ and $T{>}0$, the relation
\begin{equation} \label{eq:martingale}
\lim_{N\to 0}\P\left(\sup_{0\leq t\leq T} \frac{|M_V^N(t)|}{N}\geq \eta \right)=0.
\end{equation}
holds.
Denote by $w_{f,T}$ the modulus of continuity of a function $(f(t))$ on $[0,T]$, i.e., for $\delta{>}0$
\[
w_{f,T}(\delta)=\sup\left(|f(t){-}f(s)|: 0{\leq}s{\leq}t{\leq}T,|t-s|{\leq}\delta\right).
\]
By using again that $(\overline{Y}_*^N(t))$, $(\overline{Y}^{N}(t))$ and  $(\overline{Z}^{N}(t))$ are bounded and by Relation~\eqref{eq:martingale}, Equations~\eqref{eq:integr-ystar},~\eqref{eq:integr-y} and~\eqref{eq:integr-z} show that for any $\eps{>}0$ and $\eta{>}0$, there exist $N_0{\geq} 1$ and $\delta_0{>}0$ such that if $N{\geq} N_0$ and $\delta{<}\delta_0$ then
\[
\P\left(w_{V,T}(\delta)\geq \eta\right)\leq \eps,\quad V{\in}\left\{\overline{Y}_*^{N}, \overline{Y}^{N}, \overline{Z}^{N}\right\}.
\]
One concludes with  Theorem~15.1 of Billingsley~\cite{Billingsley}.
\end{proof}

\section{Analysis of Auxiliary Processes} \label{sec:auxiliary-first}
To study the asymptotic evolution of blocked customers, it is convenient to introduce two  important stochastic processes. The first one  describes the behavior of the  system when the second level is permanently full, and the second one corresponds to the situation when there are no blocked class~$0$ customers at level $1$.

\subsection{A Process with Saturation of Level~2} \label{sec:auxiliary-saturation}
The corresponding process is denoted by $(Y_{a*}^{N}(t),Y_a^{N}(t))$, it describes a system when level $2$ is always saturated by class~$0$ jobs. The process $(Y_{a}^N(t))$ [resp.\ $(Y_{a*}^N(t))$] indicates the number of class~$0$ jobs [resp.\ blocked] at level~$1$. For this system blocked class $0$ jobs are served at rate $\mu_{02}C_2^N$, otherwise the statistical assumptions are the same as before.

This is a Markov process with transition rates defined by
\begin{equation}\label{eq:rates-ystar}
(y_*,y)\mapsto
    \begin{cases}
    (y_*{+}1,y{-}1) &\text{at rate } \mu_{01}y,\\
    (y_*{-}1,y) &\text{"\phantom{at rate}} \displaystyle (1-p)\mu_{02}C_2^N\ind{y_*>0},\\
    (y_*{-}1,y{+}1) &\text{"\phantom{at rate}} \displaystyle p \mu_{02}C_2^N\ind{y_*>0}, \\
    (y_*, y{+}1) &\text{"\phantom{at rate}} \displaystyle p \mu_{11}(N{-}y_*{-}y).
    \end{cases}
\end{equation}
The first transition is for a $0$ job being blocked after its service at level $1$. The second one corresponds to a $0$ job leaving level $2$ allowing a blocked $0$ job to go to level $2$ and a new $1$ job is added at level $1$. The third transition is similar except that a new $0$ job enters level $1$. The last transition corresponds to a $1$ job leaving level $1$ allowing a $0$ job to enter level $1$.

As long as $Y_*^N(t){>}0$, this Markov process has the same transition rates as the process $(Y_{*}^{N}(t),Y^{N}(t))$, see Relation~\eqref{eq:Q_rates}.
\begin{proposition}\label{thm:blocY}
If the initial condition of $(Y_{a*}^N(t),Y_{a}^{N}(t))$ is such that
\begin{equation}\label{initaux}
\lim_{N\to+\infty} \frac{1}{N} (Y_{a*}^N(0),Y_{a}^{N}(0))=(y_{a*}^0,y_a^0)\in[0,1]^2,
\end{equation}
with $0 \leq y_{a*}^0{+}y_a^0 \leq 1$
then, for the convergence in distribution, the relation
\[
\lim_{N\to+\infty} \frac{1}{N} (Y_{a*}^N(t),Y_{a}^{N}(t))=(y_{a*}(t),y_a(t))
\]
holds, where $(y_{a*}(t),y_a(t))$ is a couple of continuous functions such that
\begin{equation}\label{Skoh}
y_{a*}(t)+y_a(t)\geq \overline{h}(t) \coloneq (y_{a*}^0+y_a^0)e^{-p\mu_{11}t}+\left(1-\frac{(1-p)\mu_{02}r}{p\mu_{11}}\right)\left(1-e^{-p\mu_{11}t}\right)
\end{equation}
and $(y_{a*}(t),u(t))$ is the unique solution of the following Skorokhod problem
\begin{equation}\label{Sko}
y_{a*}(t)= y_{a*}^0+\mu_{01}\int_0^t y_a(s) \diff s-\mu_{02}r t+u(t)
\end{equation}
where $(u(t))$ is a non-decreasing continuous function such that $u(0)=0$ and
\[
\int_0^{+\infty}y_{a*}(s)\,\diff u(s)=0.
\]
\end{proposition}
Concerning the Skorokhod problem in dimension~1, see Skorokhod~\cite{Skorokhod}, Chaleyat-Maurel and El~Karoui~\cite{Elkaroui}. The main trick is to express the couple $(y_{a*}(t),u(t))$ of Equation~\eqref{Sko} as a regular functional of the free process
\[
\left(y_{a*}^0{+}\mu_{01}\int_0^t y_a(s) \diff s{-}\mu_{02}r t\right).
\]
Note that, in our case, this free process  depends  on $(y_{a*}(t),y_a(t))$. 

\begin{proof}
We will proceed as follows, first show that any limiting point $(y_{a*}(t),y_a(t))$ of $(Y_{a*}^N(t),Y_{a}^{N}(t))$ is such that $(y_{a*}(t))$ can be seen as the first coordinate of the solution of a Skorokhod problem associated with a free process. In a  second step, we will show that the later process   can be expressed as a regular functional of $(y_{a*}(t))$.  One has then to use uniqueness results of Anderson and Orey~\cite{Anderson} to conclude the proof.

From the transition rates~\eqref{eq:rates-ystar}, the process $(Y_{a*}^N(t),Y_{a}^{N}(t))$ can be seen as the solution of the stochastic differential equations (SDE)
\begin{align*}
\diff Y_{a*}^N(t) &=\sum_{i=1}^{+\infty}  \ind{i\leq Y_{a}^N(t-)}\markedN{\mu_{01},i}(\diff t{,}\{0,1\})
- \ind{Y_{a*}^N(t-)>0,i\leq C_2^N}\markedN{\mu_{02},i}(\diff t{,}\{0,1\}),\\
\diff Y_{a}^{N}(t) &=-\sum_{i=1}^{+\infty}  \ind{i\leq Y_{a}^N(t-)} \markedN{\mu_{01},i}(\diff t{,}\{0,1\})
+ \ind{Y_{a*}^N(t-)>0,i\leq C_2^N}\markedN{\mu_{02},i}(\diff t{,}\{0\})\\ &\qquad +\sum_{i=1}^{+\infty}  \ind{i\leq N-Y_{a*}^N(t-)-Y_{a*}^N(t-)}\markedN{\mu_{11},i}(\diff t{,}\{0\}).
\end{align*}
With the notation
\[
\left(\overline{Y}_{a*}^N(t),\overline{Y}_a^{N}(t)\right)=
\frac{1}{N}\left(Y_{a*}^N(t),Y_{a}^{N}(t)\right),
\]
by integrating the above SDE, one gets the relations
\begin{align}
\overline{Y}_{a*}^N(t) &=\overline{Y}_{a*}^N(0) + \mu_{01}  \int_0^t\overline{Y}_a^{N}(s) \diff s
-\mu_{02} \frac{C_2^N}{N}\int_0^t \ind{\overline{Y}_{a*}^N(s)>0}\diff s+M_*^N(t)\label{Meq1}\\
\overline{Y}_a^{N}(t) &= \overline{Y}_a^{N}(0)-\mu_{01}\int_0^t \overline{Y}_a^{N}(s)\diff s
+p\mu_{11}\int_0^t \left(1-\overline{Y}_{a*}^N(s)-\overline{Y}_a^{N}(s)\right)\diff s\label{Meq2}\\
&+p\mu_{02}\frac{C_2^N}{N}\int_0^t\ind{\overline{Y}_{a*}^N(s)>0}\diff s+M^N(t),\notag
\end{align}
where $(M_*^N(t))$ and $(M^N(t))$ are local martingales.
In the same way as in the proof of Proposition~\ref{proptight} of Section~\ref{sec:init},
one can prove that the sequence of processes $(\overline{Y}_{a*}^N(t), \overline{Y}_a^{N}(t))$ is tight and that any of its limiting points is a continuous process.

Let $(y_{a*}(t),y_a(t))$ be a limiting point, i.e., for some subsequence $(N_k)$ the relation
\[
\lim_{k\to+\infty} \left(\overline{Y}_{a*}^{N_k}(t), \overline{Y}_a^{N_k}(t)\right)=(y_{a*}(t),y_a(t))
\]
holds for the convergence in distribution of processes.
Denote
\begin{equation}\label{eqX1}
  F_a^N(t)=\overline{Y}_{a*}^N(0) + \mu_{01}  \int_0^t\overline{Y}_a^{N}(s) \diff s
  -\mu_{02}\frac{C_2^N}{N}t +M_*^N(t).
\end{equation}
Equation~\eqref{Meq1} can be written as
\[
\overline{Y}_{a*}^N(t) =F_a^N(t)+\mu_{02}\frac{C_2^N}{N} \int_0^t \ind{\overline{Y}_{a*}^N(s)=0}\diff s,
\]
so that the couple
\[
\left(\overline{Y}_{a*}^N(t), \mu_{02}\frac{C_2^N}{N} \int_0^t \ind{\overline{Y}_{a*}^N(s)=0}\diff s\right)
\]
is the solution of the Skorokhod problem associated with the free process $(F_a^N(t))$. See Skorokhod~\cite{Skorokhod} and Appendix~D of Robert~\cite{Robert} for a brief account.

For the convergence in distribution of processes, one has
\begin{equation}\label{s1}
\lim_{k\to+\infty} (F_a^{N_k}(t))=(f_a(t)) \coloneq \left(y_{a*}^0+\mu_{01}\int_0^t y_a(s) \diff s
-\mu_{02}r t\right),
\end{equation}
since, as before, the martingales are vanishing as $N$ gets large. 
From Proposition~D.4 of the appendix of Robert~\cite{Robert}, one gets that $(y_{a*}(t))$ is the first coordinate of the solution of the Skorokhod problem associated with $(f_a(t))$ and $(y_{a*}(t))$ is differentiable almost everywhere for the Lebesgue measure on $\R_+$. In particular Relation~\eqref{Sko} holds. 

Since the free process $(f_a(t))$ depends on $(y_{a*}(t),y_a(t))$, there is no guarantee of the uniqueness of such a limit point $(y_{a*}(t),y_a(t))$. We now give a representation of $(f_a(t))$ in terms of $(y_{a*}(t))$. 
We proceed by getting rid of the process $(\overline{Y}_a^N(t))$ in the expression~\eqref{eqX1} of $(F_a^N(t))$.
From Equations~\eqref{Meq1} and~\eqref{Meq2}, we get  the relation
\[
p \overline{Y}_{a*}^N(t) + \overline{Y}_a^N(t) =
\begin{multlined}[t]
p \overline{Y}_{a*}^N(0) + \overline{Y}_a^N(0) + p \mu_{11} t - p \mu_{11} \int_0^t \overline{Y}_{a*}^N(s) \diff s\\
- ((1-p) \mu_{01} + p \mu_{11}) \int_0^t \overline{Y}_a^N(s) \diff s + p M_*^N(t) + M^N(t).
\end{multlined}
\] 
 By reordering the terms, one gets the relation
\begin{equation}\label{eqax1}
    \overline{Y}_a^N(t) + \overline{\mu} \int_0^t \overline{Y}_a^N(s) \diff s =
    \begin{multlined}[t]
    \left(p \overline{Y}_{a*}^N(0) + \overline{Y}_a^N(0)\right) + p \mu_{11} t - p \overline{Y}_{a*}^N(t)\\
    - p \mu_{11} \int_0^t \overline{Y}_{a*}^N(s) \diff s + p M_*^N(t) + M^N(t),
    \end{multlined}
\end{equation}
with $\overline{\mu} \steq{def} (1{-}p) \mu_{01} {+} p \mu_{11}$.
Hence, by denoting
\[
  K^N(t)\steq{def} \frac{1}{\overline{\mu}}\left(p \overline{Y}_{a*}^N(0) + \overline{Y}_a^N(0)\right) \left(e^{\overline{\mu} t} - 1\right) + \dfrac{p \mu_{11}}{\overline{\mu}^2} \left(1+(\overline{\mu}t-1)e^{\overline{\mu} t} \right),
  \]
  and $(k(t))$  its limit,
  \[
  k(t)\steq{def} \frac{1}{\overline{\mu}}\left(p y_{a*}^0 +y_a^0\right) \left(e^{\overline{\mu} t} - 1\right) + \dfrac{p \mu_{11}}{\overline{\mu}^2} \left(1+(\overline{\mu}t-1)e^{\overline{\mu} t} \right) ,
  \]
from Equation~\eqref{eqax1}, trite calculations give the representation
\begin{multline*}
\int_0^t \overline{Y}_a^N(s) \diff s = K^N(t) e^{-\overline{\mu} t} - p \int_0^t \left(\overline{Y}_{a*}^N(s) + \mu_{11} \int_0^s \overline{Y}_{a*}^N(u) \diff u \right)e^{-\overline{\mu} (t-s)} \diff s\\
+ \int_0^t \left(p M_*^N(s) + M^N(s)\right)e^{-\overline{\mu}(t- s)} \diff s.
\end{multline*}
Therefore, we can write the free process $(F_a^N(t))$ as
\begin{multline}\label{free1}
F_a^N(t)=G\left(\overline{Y}_{a*}^N\right)(t)+ \overline{Y}_{a*}^N(0)+\mu_{01} K^N(t) e^{-\overline{\mu} t}-\mu_{02}\dfrac{C_2^N}{N} t\\ +\mu_{01}\int_0^t \left(p M_*^N(s) + M^N(s)\right)e^{-\overline{\mu}(t-s)} \diff s + M_*^N(t),
\end{multline}
where $G(\cdot)$ is a functional on Borelian functions $(x(t))$ defined by
\[
G(x)(t) = - p \mu_{01}   \int_0^t \left(x(s) + \mu_{11} \int_0^s x(u) \diff u \right)e^{-\overline{\mu}(t-s)} \diff s.
\]
This gives us an alternative representation of $(f_a(t))$ as
\begin{equation}\label{s2}
f_a(t)=\overline{G}(y_{a*})(t)\steq{def}G(y_{a*})(t)+y_{a*}^0+\mu_{01} k(t) e^{-\overline{\mu} t}-\mu_{02}r t.
\end{equation}
We have shown that $(y_{a*}(t))$ is the first coordinate of $(y_{a*}(t),u(t))$, the solution of a generalized Skorokhod problem associated to the functional $\overline{G}$,
\[
y_{a*}(t)=\overline{G}(y_{a*})(t)+u(t) \text{ and } \int_0^{+\infty} y_{a*}(s)\diff u(s)=0,
\]
with the usual assumptions on $(y_{a*}(t))$ and $(u(t))$.  See  Anderson and Orey~\cite{Anderson}. 
For any Borelian functions $(a(t))$ and $(b(t))$ on $\R_+$, it is not difficult to check that
\[
\| \overline{G}(a) - \overline{G}(b) \|_{\infty,t}\steq{def}   \sup_{0\leq s\leq t} \| \overline{G}(a)(s) - \overline{G}(b)(s) \| \leq C_t \int_0^t \| a - b \|_{\infty,s} \diff s,
\]
with $C_t = p \mu_{01} (1 + \mu_{11}t)$.
Anderson and Orey~\cite{Anderson} show that such $(y_{a*}(t))$ is unique.
The convergence in distribution follows:
\[
\lim_{N\to+\infty} \left(\overline{Y}_{a*}^N(t), \mu_{02}\frac{C_2^N}{N} \int_0^t \ind{\overline{Y}_{a*}^N(s)=0}\diff s\right)=(y_{a*}(t),u(t)). 
\]
Consequently, Relations~\eqref{Meq1} and~\eqref{Meq2}  give the relations
\begin{equation}\label{Sk01}
\begin{cases}
\displaystyle y_{a*}(t) =y_{a*}^0+\mu_{01}\int_0^t y_a(s) \diff s-\mu_{02}r t+u(t),\\
\displaystyle y_a(t)= y_a^0-\mu_{01}\int_0^t y_a(s)\diff s+ p\mu_{11}\int_0^t \left(1-y_{a*}(s)-y_a(s)\right)\diff s\\
\ \hfill + p\mu_{02}r t-p u(t). 
\end{cases}
\end{equation}
By using Relations~\eqref{s1} and~\eqref{s2}, one deduce the uniqueness of $(y_a(t))$and, therefore, the convergence in distribution of the sequence of processes  $(Y_{a*}^N(t),Y_{a}^{N}(t))$. 

We now prove that the limit $(y_{a*}(t),y_a(t))$  satisfies necessarily $y_{a*}(t){+} y_a(t){\geq} \overline{h}(t)$ for all $t$, where $\overline{h}$ is the solution of
\[\overline{h}(t) = (y_{a*}^0+y_a^0)-(1{-}p)\mu_{02}r t+p\mu_{11}\int_0^t \left(1-\overline{h}(s)\right)\diff s,
\]
that is,
\[
  \overline{h}(t) = (y_{a*}^0+y_a^0) e^{-p \mu_{11} t} + \left(1 - \dfrac{(1-p) \mu_{02} r}{p \mu_{11}}\right) (1 - e^{-p \mu_{11} t}).
\]
First note that, for any $N$, the process $\overline{Y}_a^N$ is bounded above by $1$, so that $F_a^N$ is Lipschitz. Hence, again by Proposition~D.4 of the appendix of Robert~\cite{Robert}, $u$ is  also Lipschitz, and thus continuous.

From Relations~\eqref{Sk01}, one gets that the identity
\begin{multline*}
(y_{a*}(t)+y_a(t)) =(y_{a*}^0+y_a^0)-(1{-}p)\mu_{02}r t\\ + p\mu_{11} \int_0^t \left(1-(y_{a*}(s) + y_a(s))\right)\diff s + (1{-}p) u(t)
\end{multline*}
holds, so that the difference $y_{a*}{+}y_a{-} \overline{h}$ satisfies the system
\begin{equation*}
    x(t) + p \mu_{11} \int_0^t x(s) \diff s = (1{-}p)u(t), \text{ with } x(0)=0
\end{equation*}
Any continuous solution $(x(t))$ of this system is non-negative.  Suppose that there exists $t_1{>} 0$ such that $x(t_1) {<} 0$.
Then, by continuity of $(x(t))$, there exists $t_0 {<} t_1$ such that $x(t_0){=} 0$ and $x(t) {<} 0$ for $t_0{<} t{<}t_1$.
But
\[
x(t_1) - x(t_0) = x(t_1) = (1{-}p)(u(t_1) - u(t_0)) - p \mu_{11} \int_{t_0}^{t_1} x(s) \diff s
\]
and the right-hand-side of the equality is positive because $(u(t))$ is non-decreasing and $(x(t))$ is negative on this interval, which is a contradiction.
Relation~\eqref{Skoh} is established. The proposition is proved.
\end{proof}
\begin{proposition}  \label{thm:diff_first_regime}
Under Condition~\eqref{eq:cond1},
there exists  $t_0{\geq} 0$, independent of the initial state~\eqref{initaux}, such that for $t{\geq} t_0$, the functions $(y_{a*}(t))$ and $(y_a(t))$ introduced in Proposition~\ref{thm:blocY} are differentiable at $t$ and
\begin{equation}\label{eq:diff_first_regime}
\begin{cases}
 \dfrac{\diff}{\diff t}{y}_{a*}(t)=\mu_{01}y_a(t) - \mu_{02}r,\\ \\
 \dfrac{\diff}{\diff t}{y}_a(t)=-(\mu_{01}{+}p\mu_{11})y_a(t)-p\mu_{11}y_{a*}(t)+p(\mu_{02}r{+}\mu_{11})
\end{cases}
\end{equation}
 \end{proposition}
Any solution $(y_{a*}(t),y_{a}(t))$ of the differential system~\eqref{eq:diff_first_regime} converges to 
\begin{equation} \label{eq:limit_first_regime}
\left(1 {-} \frac{\mu_{02} r}{\mu_{01}}\left( \frac{(1{-}p)\mu_{01}}{p \mu_{11}} {+}1\right),\frac{\mu_{02} r}{\mu_{01}}\right).
\end{equation}
\begin{proof}
The above proposition shows that
\[
\liminf_{t\to+\infty} y_{a*}(t)+y_a(t) \geq 1-\frac{(1-p)\mu_{02}r}{p\mu_{11}}.
\]

Let
\[
\eps_0= \left(1-\frac{(1-p)\mu_{02}r}{p\mu_{11}}\right)-\frac{\mu_{02}r}{\mu_{01}},
\]
then $\eps_0{>}0$ by Condition~\eqref{eq:cond1}. Let $t_0$ be such that if $t{\geq} t_0$ then
\begin{equation}\label{abc}
y_{a*}(t)+y_a(t) \geq 1-\frac{(1-p)\mu_{02}r}{p\mu_{11}}- \frac{\eps_0}{2}.
\end{equation}
The classical representation of the solution of one-dimensional Skorokhod problem, see Relation~(D.1) p.376 of Robert~\cite{Robert}, gives the identity
\[
y_{a*}(t)= \left(y_{a*}^0+\mu_{01}\int_0^t y_a(u) \diff u-\mu_{02}r t\right) \vee
\sup_{0\leq s\leq t}\left(\mu_{01}\int_s^t y_a(u) \diff u-\mu_{02}r (t-s)\right).
\]
If $t_1{>}t_0$ is such that $y_{a*}(t_1){=}0$, then by continuity of $(y_a(t))$, one gets the relation
\[
\mu_{01}y_a(t_1)\leq \mu_{02}r,
\]
and, by using Relation~\eqref{abc},
\[
1-\frac{(1-p)\mu_{02}r}{p\mu_{11}}-\frac{\eps_0}{2}\leq y_{a*}(t_1)+y_a(t_1)=y_a(t_1)\leq \frac{\mu_{02}r}{\mu_{01}}
\]
which leads to a contradiction. One concludes that $t\mapsto y_{a*}(t)$ is positive for $t{>}t_0$ and consequently that the measure $\diff u(t)$ vanishes on the interval $(t_0,+\infty)$. The proposition is proved.
\end{proof}

\subsection{A System without blocked jobs}\label{sec:auxiliary-free}
A second  auxiliary process is introduced, it is denoted by $(Y_b^N(t),Z_b^N(t))$.  It describes the situation when there are no blocked jobs at level one: if a class~0 job finishes while level two is saturated, i.e., $Z_b^N(t){=}0$, then it leaves the system, instead of being blocked.  If there are free servers at level two, the process behaves in the same way as the main process under study.  The process $(Y_b^N(t), Z_b^N(t))$ is a Markov process, with the following transition rates:
\begin{equation}\label{eq:rates-z}
(y,z)\mapsto
    \begin{cases}
    (y{-}1,0) &\text{at rate \phantom{"}} (1{-}p) \mu_{01} y \ind{z = 0},\\
    (y{-}1,z{-}1) &\text{"\phantom{at rate} } \displaystyle (1{-}p)\mu_{01} y \ind{z >0},\\
    (y,z{-}1) &\text{"\phantom{at rate} } \displaystyle p \mu_{01} y \ind{z >0},\\
    (y{+}1, z) &\text{"\phantom{at rate} } \displaystyle p \mu_{11}(N{-}y),\\
    (y, z{+}1) &\text{"\phantom{at rate} } \displaystyle \mu_{02} (C_2^N {-}z).
    \end{cases}
\end{equation}
Note that, when $z{>}0$, this Markov process has the same transition rates as the process $(X^N(t))$ (see \eqref{eq:Q_rates}).

\begin{proposition} \label{thm:blocZ}
  If the initial condition of $(Y_b^{N}(t), Z_b^N(t))$ is such that
\begin{equation}\label{initaux2}
\lim_{N\to+\infty} \frac{1}{N} (Y_b^{N}(0),Z_b^{N}(0))=(y_b^0,z_b^0)\in[0,1]\times[0,r],
\end{equation}
then, for the convergence in distribution, the relation
\[
\lim_{N\to+\infty} \frac{1}{N} (Y_b^{N}(t),Z_b^{N}(t))=(y_b(t),z_b(t))
\]
holds, where $(y_b(t))$ is given by 
\begin{equation} \label{eq:y-second-regime}
    y_b(t) = y_b^0 e^{-(p \mu_{11} {+} (1{-}p) \mu_{01})t} + \dfrac{p \mu_{11}}{p \mu_{11} {+} (1{-}p) \mu_{01}} \left(1{-} e^{-(p \mu_{11} {+} (1{-}p) \mu_{01})t} \right)
\end{equation}
and $(z_b(t))$ is the unique solution of the Skorokhod problem
\begin{equation}\label{eq:Sko2}
z_b(t) = z_b^0 + \mu_{02} r t - \mu_{02} \int_0^t z_b(s) \diff s - \mu_{01} \int_0^t y_b(s) \diff s + u(t)
\end{equation}
where $(u(t))$ is a non-decreasing continuous function such that $u(0){=}0$ and
\[
\int_0^{+\infty}z_b(s)\diff u(s)=0\,.
\]
\end{proposition}
As before, the free process associated to the Skorokhod problem is 
\[
\left(z_b^0 + \mu_{02} r t - \mu_{02} \int_0^t z_b(s) \diff s - \mu_{01} \int_0^t y_b(s) \diff s\right), 
\]
it is also a functional of $(z_b(t))$, see the proof of Proposition~\ref{thm:blocY}.
\begin{proof}
From the transition rates~\eqref{eq:rates-z} and as in the proof of Proposition~\ref{thm:blocY},  if
\[
\left(\overline{Y}_b^N(t),\overline{Z}_b^N(t)\right) \coloneq
\dfrac{1}{N}\left({Y}_b^N(t),{Z}_b^N(t)\right),
\]
then  one gets the evolution equations
\[
\overline{Y}_b^{N}(t) =\overline{Y}_b^{N}(0)-\mu_{01}(1-p) \int_0^t \overline{Y}_b^N(s) \diff s
+ p\mu_{11}\int_0^t(1-\overline{Y}_b^{N}(s))\diff s + M_Y^N(t),
\]
  \begin{multline*}
    \overline{Z}_b^{N}(t) =\overline{Z}_b^{N}(0) + \mu_{02}\frac{ C_2^N}{N} t - \mu_{02}\int_0^t\overline{Z}_b^N(s)\diff s
\\-\mu_{01}\int_0^t \overline{Y}_b^N(s) \diff s + M_Z^N(t)
+ R^N_Z(t),
  \end{multline*}
  with
  \[
  R^N_Z(t) = \mu_{01} \int_0^t \overline{Y}_b^N(s)\ind{\overline{Z}_b^N(s)=0}\diff s,
  \]
  and $(M_Z^N(t))$ and $(M_Y^N(t))$ are local martingales. It is easily seen that these  two martingales vanish when $N$ gets large and hence that, with the criterion of the modulus of continuity,  the sequence of processes $(\overline{Y}_b^{N}(t))$ is tight. Furthermore, any limiting point  $({y}_b(t))$ satisfies the integral equation
\begin{equation*}
    y_b(t) = y_b^0 - (p \mu_{11} + (1-p) \mu_{01}) \int_0^t y_b(s) \diff s + p \mu_{11} t,
\end{equation*}
so that Relation~\eqref {eq:y-second-regime} holds. 

Clearly, $(\overline{Z}_b^N(t),{R}^N_Z(t))$ is the solution of a generalized Skorokhod process associated with  a free process which  depends itself on $\overline{Z}_b^N(t)$).
One concludes in the same way as in the proof of  Proposition~\ref{thm:blocY}.
\end{proof}
We now prove that,  under Condition~\eqref{eq:cond2}, the reflecting part  of the Skorokhod problem  of the last proposition vanishes for $t$ large enough.
\begin{proposition}\label{prop5}
Under  Condition~\eqref{eq:cond2}, there exists $t_0{\geq} 0$, independent of the initial state~\eqref{initaux2}, such that for $t{\geq}t_0$,
 \begin{equation} \label{eq:z-second-regime}
    \dfrac{\diff}{\diff t} z_b(t) = \mu_{02} r - \mu_{02} z_b(t) - \mu_{01} y_b(t).
  \end{equation}
Furthermore,
\[
\lim_{t\to+\infty} (y_b(t), z_b(t)) =\left(\frac{p \mu_{11}}{p \mu_{11} {+} (1{-}p) \mu_{01}},  r {-} \frac{p \mu_{01}\mu_{11}}{\mu_{02}(p \mu_{11} {+} (1{-}p) \mu_{01})} \right).
\]
\end{proposition}
\begin{proof}
By Proposition~\ref{thm:blocZ}, $(y_a(t))$ converges to $\overline{y}_b$ as $t$ gets large.
By Condition~\eqref{eq:cond2}, the relation  $\overline{y}_{b} {<} \mu_{02} r / \mu_{01}$ holds.
Consequently, there exists $t_0$ such that the inequality
\[
 \mu_{01} y_b(t){<} \mu_{02} r
\]
holds  for $t{\geq}t_0$.

Theorem~D.1 of Robert~\cite{Robert} gives the relation
\begin{multline*}
  z_b(t) = \left(z_b^0 + \mu_{02} r t - \mu_{02} \int_0^t z_b(s) \diff s - \mu_{01} \int_0^t y_b(s) \diff s \right) \\\vee \sup_{0 \leq s \leq t} \left( \mu_{02} r (t{-}s) - \mu_{02} \int_s^t z_b(u) \diff u - \mu_{01} \int_s^t y_b(u) \diff u \right).
\end{multline*}
Suppose that $z_b(t_1){=} 0$ for some $t_1 {>} t_0$, then in particular, for $0{\leq}s{\leq} t_1$,
\[
\mu_{02} r (t_1{-}s) - \mu_{02} \int_s^{t_1} z_b(u) \diff u - \mu_{01} \int_s^{t_1} y_b(u) \diff u\leq 0,
\]
by continuity of $(z_b(t))$ and $(y_b(t))$, this gives that
\[
\mu_{02} r{-}\mu_{01} y_b(t_1)= \mu_{02} r{-}\mu_{02} z_b(t_1){-}\mu_{01} y_b(t_1) \leq 0,
\]
contradiction.
This implies that $z_b(t){>} 0$ holds  for $t {\geq} t_0$ and, consequently, the measure $\diff u(t)$ vanishes on $[t_0,{+}\infty)$. The proposition is proved.
\end{proof}

\section{Asymptotic Study of the Blocking Phenomenon} \label{sec:first}
The goal of this section is of showing that if the ratio $r{\sim}C_2^N/C_1^N$ of the capacities of the two levels of our system is less than the quantity 
\[
r_c\steq{def} \frac{p}{\mu_{02}}\left/\left(\frac{p}{\mu_{01}}{+}\frac{1{-}p}{\mu_{11}}\right)\right.,
\]
then there exists some fixed instant $t_0$ such that, with a probability converging to $1$ as $N$ gets large, the number of blocked servers at level~$1$ is of the order of $N$ on any finite time interval after $t_0$. Otherwise, if $r{>}r_c$, then there exists $t_0{>}0$ such that the number of blocked servers is $0$ with high probability on any finite time interval after $t_0$. These results are respectively proved in Sections~\ref{satsec} and~\ref{undersec}. The proofs use the technical tools introduced in the last section and additional probabilistic arguments. 
\subsection{The Overloaded Regime}\label{satsec}
\addcontentsline{toc}{section}{\hspace{2mm}\thesubsection. The Overloaded Regime}
In this section, we will assume that Condition~\eqref{eq:cond1} holds, i.e., that $r{<}r_c$. 
Recall that $(Y_*^N(t),Y^N(t),Z^N(t))$ describes the state of our system.  The following proposition is a technical result that shows that type $0$ jobs occupy at least a  fixed fraction of the first level after some time. 

\begin{proposition} \label{thm:yplusystar}
Under Condition~\eqref{eq:cond1},  for $\eps{>} 0$, there exists $t_0{>}0$ such that, for any initial fluid state~\eqref{ScalingInit}, and  $T{\geq}t_0$, then
    \[
      \lim_{N \to \infty} \P \left( \inf_{t \in [t_0, T]} \big(Y_*^N(t){+}Y^{N}(t)\big) \geq N( \bar{y}{-}\eps )\right) = 1,
    \]
    with
    \[
    \bar{y} \coloneq \frac{p \mu_{11}}{ p \mu_{11} {+} (1{-}p) \mu_{01}}.
    \]
\end{proposition}

\begin{proof}
Define, for $t\geq 0$,
\[
H^N(t)=N\bar{y}{-}Y_*^N(t){-}Y^N(t),
\]
we want to show that $H^N(t)$ is, with high probability, below $N\eps$ on any finite time interval after some finite fixed instant. 

The stochastic differential equations~\eqref{eq:integr-ystar-diff} and~\eqref{eq:integr-y-diff} give the relation
\begin{equation} \label{eq:yplusystar}
\diff H^N(t) =  -\sum_{i=1}^{+\infty}\ind{i\leq N(1-\bar{y})+H^{N}(t-)}\markedN{\mu_{11},i}(\diff t{,}\{0\}) +\diff D^N_0(t)
\end{equation}
where  $(D^N_0(t))$ is the process  associated with the positive jumps of this SDE,
\begin{multline*}
D^N_0(t)\steq{def} \sum_{i=1}^{+\infty}\int_0^t \ind{Y_*^N(s)>0,i\leq C_2^N}\markedN{\mu_{02},i}(\diff s{,}\{1\})
\\  +\ind{i\leq Y^N(s),Z^N(s)>0} \markedN{\mu_{01},i}(\diff s{,}\{1\}).
\end{multline*}
In the following we will use repeatedly, without mentioning it explicitly,  the coupling of Poisson process with ordered rates described at the beginning of Section~\ref{sec:init}.

If, for some $T{>}0$, one has that, for all $0{\leq} t{\leq} T$,  $H^N(t){\geq} \lceil \eps N\rceil $, then, on this time interval, $Y^N(t){\leq}N(\bar{y}{-}\eps)$ and, $D^N_0(t){-}D^N_0(s){\leq}D^N_1(t){-}D^N_1(s)$ for all $0{\leq}s{\leq}t$, with
\begin{multline*}
D^N_1(t)\steq{def}\int_0^t\sum_{i=1}^{+\infty} \ind{Y_*^N(s)>0,i\leq C_2^N}\markedN{\mu_{02},i}(\diff s{,}\{1\})
\\\phantom{\sum}+\ind{i\leq N(\bar{y}-\eps),Z^N(s)>0} \markedN{\mu_{01},i}(\diff s{,}\{1\}).
\end{multline*}
By using classical results on superposition and thinning of independent Poisson processes, see Kingman~\cite{Kingman} for example, we have $(D^N_1(t)) \steq{dist} (D^N_2(t))$, with
\[
D^N_2(t)\steq{def} \int_0^t \ind{Y_*^N(s)>0}{\cal N}_{\mu_{02}(1-p)C_2^N}(\diff s)
{+}\ind{Z^N(s)>0} {\cal N}_{\mu_{01}(1-p)N(\bar{y}-\eps)}(\diff s).
\]
It is easily seen that Condition~\eqref{eq:cond1} is equivalent to the relation $r\mu_{02}{<}\bar{y}\mu_{01}$, hence for $\eta{>}0$ there exists some $N_0$ such that
\[
\mu_{02}C_2^N<\mu_{01}(\bar{y}+\eta)N, \quad\text{ for } N{\geq}N_0.
\]
By using this inequality and the relation $\{Y_*^N(t-){>}0\}{\subset}\{Z^N(t-){=}0\}$,  one has
\[
  \int_0^t \ind{Y_*^N(s)>0}{\cal N}_{\mu_{02}(1-p)C_2^N}(\diff s) \leq
  \int_0^t \ind{Z^N(s)=0}{\cal N}_{\mu_{01}(1-p)N(\bar{y}+\eta)}(\diff s).
\]
Moreover, since
\[
  \int_0^t \ind{Z^N(s)>0} {\cal N}_{\mu_{01}(1-p)N(\bar{y}-\eps)}(\diff s) \leq
  \int_0^t \ind{Z^N(s)>0}{\cal N}_{\mu_{01}(1-p)N(\bar{y}+\eta)}(\diff s),
\]
hence $D^N_2(t){-}D^N_2(s){\leq}D^N_3(t){-}D^N_3(s)$, for $0{\leq}s{\leq}t$, with
\begin{equation*}
  D^N_3(t) \steq{def}  \int_0^t \ind{Z^N(s)=0}{\cal N}_{\mu_{01}(1-p)N(\bar{y}+\eta)}(\diff s)
+\ind{Z^N(s)>0} {\cal N}_{\mu_{01}(1-p)N(\bar{y}+\eta),2}(\diff s),
\end{equation*}
where 
\(
  {\cal N}_{\mu_{01}(1-p)N(\bar{y}+\eta)} \text{ and }{\cal N}_{\mu_{01}(1-p)N(\eps+\eta),2}
\)
are two independent Poisson point processes.
The integer valued  process $(D^N_3(t))$ has jumps of size $1$ and it is easily checked that
\[
\left(D^N_3(t){-}\mu_{01}N(1{-}p)(\bar{y}{+}\eta)t\right)
\]
is a martingale. From Watanabe's Theorem, see Watanabe~\cite{Watanabe}, one gets that $(D^N_3(t))$ is a Poisson process on $\R_+$ with rate $\lambda_H N$ with $\lambda_H \steq{def} \mu_{01}(1{-}p)(\bar{y}{+}\eta)$.

If, for all $0{\leq} t{\leq} T$,  $H^N(t){\geq} \lceil \eps N\rceil $ then
\begin{equation}\label{eqp1}
D^N_0(t){-}D^N_0(s)\leq D^N_3(t){-}D^N_3(s) \text{  for } 0{\leq}s{\leq}t,
\end{equation}
where $(D^N_3(t))$ is a Poisson process with rate $\lambda_H N$.

Assume that $H^N(0){\geq}\lceil \eps N\rceil $ and define  the process $(\overline{H}^N\!\!(t))$, with the initial condition $\overline{H}^N\!\!(0){=}H^N(0){-}\lceil \eps N\rceil$ and such that 
\begin{equation} \label{eqaux1}
 \diff \overline{H}^N\!\!(t) =   {\cal N}_{\lambda_H N}(\diff t)-\ind{\overline{H}^N(t-)>0}\sum_{i=1}^{+\infty}\ind{i\leq N(1-\bar{y})+\eps N}\markedN{\mu_{11},i}(\diff t{,}\{0\})
\end{equation}
holds for $t{\geq}0$. 
Clearly, $(\overline{H}^N\!\!(t))$ has the same distribution as $(L_H(Nt))$, where $(L_H(t))$ is the Markov process associated with an $M/M/1$ queue whose arrival and services rates are respectively  
$\lambda_H$ and  $\mu_H \steq{def} p\mu_{11}(1{-}\bar{y}{+}\eps)$. Note that, by definition of $\bar{y}$, 
\[
\lambda_H{-}\mu_H= \mu_{01}(1{-}p)\eta-p\mu_{11}\eps.
\]
One can choose $\eta{>}0$ so that $\lambda_H{-}\mu_H{<}{-}p\mu_{11}\eps/2$,  hence $(L_H(t))$ is an ergodic Markov process in this case. 

Let $T_\eps^N{=}\inf\{t\geq 0 : H^N(t){<} \eps N\}$, in particular  $H^N(t){\geq} \lceil \eps N\rceil $   for $0{\leq} t{\leq} T_\eps^N$. 
From Relations~\eqref{eq:yplusystar}, \eqref{eqp1} and~\eqref{eqaux1}, we get therefore the inequality
\begin{equation}\label{eqaux2}
  {H}^N(t){-}\lceil \eps N\rceil\leq \overline{H}^N\!\!(t),
\end{equation}
for $t{<}T_\eps^N$,  hence  $T_\eps^N{\leq} \tau_L^N$, with
\[
\tau_L^N=\inf\left\{t>0 : \overline{H}^N(t)=0\right\}=\inf\{t>0:L_H(Nt)=0\}. 
\]
Since $\overline{H}^N\!\!(0){\leq}H^N(0){\leq} N$, Proposition~5.16 of Robert~\cite{Robert} gives that, for any $t_0$ such that $t_0{>}{1}/{(\mu_H{-}\lambda_H)}$ then 
\[
\lim_{N\to+\infty}\P\left(T_\eps^N\leq t_0\right)\geq \lim_{N\to+\infty}\P\left(\tau_L^N\leq t_0\right)=1. 
\]
By using the strong Markov property, up to a change of time origin, one can assume that $H^N(0){\leq}\lceil \eps N\rceil$. If $(\overline{H}^N\!\!(t))$ is defined as before with the initial condition  $\overline{H}^N(0){=}0$, then it is not difficult to show that  Relation~\eqref{eqaux2} holds for all $t{\geq}0$. For the excursions of $(H^N(t))$ below $\lceil \eps N\rceil$ this is clear and for the excursions above this level it has just been proved. In particular, for any $T>0$,
\begin{multline*}
 \P\left(\inf_{0\leq t\leq T} Y_*^N(t){+}Y^N(t) \leq N(\bar{y}{-}2\eps)\right)= \P\left(\sup_{0\leq t\leq T} \overline{H}^N(t)\geq 2\eps N\right)\\\leq \P_0\left(\sup_{0\leq t\leq NT} L_H(t)\geq \lfloor \eps N\rfloor\right),
\end{multline*}
and the last quantity is the probability that the hitting time of $\lfloor \eps N\rfloor$ by an $M/M/1$ queue starting from $0$ is less that $NT$. Proposition~5.11 of Robert~\cite{Robert} shows that  this hitting time is of the order of $(\mu_H/\lambda_H)^{\lfloor \eps N\rfloor}$ and therefore exponentially large in $N$ (recall that $\lambda_H{<}\mu_H$). In particular, the last term of the right-hand-side of the above  relation is converging to $0$ as $N$ gets large. The proposition is proved. 
\end{proof}
The above proof relies on the comparisons of point processes associated to the counting processes $(D_i(t))$, $i{\in}\{0,1,2,3\}$. We have, for example, that the point process associated to $(D_0(t))$ is ``smaller'' that the one associated to $(D_1(t))$:  $D^N_0(t){-}D^N_0(s){\leq}D^N_1(t){-}D^N_1(s)$ for all $0{\leq}s{\leq}t$. In the following, for  convenience, we will use the notation $D_0(\diff t){\leq} D_1(\diff t)$. 
\begin{proposition} \label{thm:ystarminusz}
Under Condition~\eqref{eq:cond1}, for any $\varepsilon{>}0$ small enough, there exists  $t_1{>}0$ such that, for any initial fluid state~\eqref{ScalingInit}, and for any $T{\geq}t_1$, the relation 
    \[
        \lim_{N \to \infty} \P \left( \inf_{t \in [t_1,T]} Y_*^N(t) \geq \eps N \right) = 1.
        \]
        holds. 
\end{proposition}
\begin{proof}
From Equations~\eqref{eq:integr-ystar-diff} and~\eqref{eq:integr-z-diff}, we get that 
\begin{equation*}
  \diff \left(Y_*^N{-} Z^N\right)(t) =\sum_{i=1}^{+\infty} \ind{i\leq Y^N(t{-})}{\cal N}_{\mu_{01},i}(\diff t)
  -\sum_{i=1}^{+\infty} \ind{i \leq C_2^N - Z^N(t{-})}{\cal N}_{\mu_{02},i}(\diff t).
\end{equation*}
By Proposition~\ref{thm:yplusystar}, there exists  $t_0$ be such that, for $T{\geq}t_0$, for the events
\[
{\cal A}_N\steq{def}\left\{ \inf_{t \in [t_0, T]}\big(Y_*^N(t) {+} Y^{N}(t)\big) \geq N( \bar{y} {-} \eps )\right\},
\]
the sequence $(\P({\cal A}_N))$ converges to $1$.

Suppose that, for some time $t {\in}(t_0,T)$,   $Y_*^N(t){<}2\eps N$, then, on the event ${\cal A}_N$, $Y^N(t){\geq}N( \bar{y}{-} 3 \eps)$ and consequently
\begin{align}
  \diff \left(Y_*^N{-} Z^N\right)(t) &\geq \sum_{i=1}^{+\infty} \ind{i\leq N( \bar{y} - 3 \eps)}{\cal N}_{\mu_{01},i}(\diff t)
  -\sum_{i=1}^{+\infty} \ind{i \leq C_2^N}{\cal N}_{\mu_{02},i}(\diff t)\label{eqY1}\\
& \steq{dist} {\cal N}_{\mu_N}(\diff t)-{\cal N}_{\lambda_N}(\diff t).\notag
\end{align}
with $\mu_N{=}\lfloor \mu_{01}N ( \bar{y}{-}3 \eps)\rfloor$ and $\lambda_N{=}\mu_{02}C_2^N$.
As noted before, Condition~\eqref{eq:cond1}  is equivalent to the relation $\bar{y}\mu_{01}{-}r\mu_{02}{>}0$. Since
\[
\lim_{N\to+\infty} \frac{\mu_N{-}\lambda_N}{N}=(\mu_{01} \bar{y} - \mu_{02} r) - \mu_{01} 3\eps,
\]
one can find  $0{<}\lambda{<}\mu$ and $\eps$ sufficiently small, such that for $N$ sufficiently large   $\lambda_N{\geq }\lambda N$ and $\mu_N{\leq}\mu N$ hold.   If
\[
T_N=\inf\{t{>}0: Y_*^N(t){\geq}2\eps N\},
\]
then since $\left(Y_*^N{-} Z^N\right)(0){\geq}{-}C_2^N$, Relation~\eqref{eqY1} gives the existence of $t_1{>}0$ such that
\[
\lim_{N\to+\infty}\P(T_N{\leq} t_1)=1.
\]
We now assume that $\left(Y_*^N{-} Z^N\right)(0){=}\lceil 2\eps N\rceil$. By taking $T{>}t_1$, using  Relation~\eqref{eqY1} and the estimates for $\lambda_N$ and $\mu_N$, we get that, on the event ${\cal A}_N$,
\[
\left(Y_*^N{-} Z^N\right)(t)\geq \lceil 2\eps N\rceil- X(Nt)
\]
holds for all $t{>}0$, where $(X(t))$ is an $M/M/1$ queue with arrival [resp.\ service] rate $\lambda$ [resp.\ $\mu$]  starting at $0$. With the same argument as the end of the previous proof, since $\lambda{<}\mu$, we have that, for any $T{>}0$,
\[
\lim_{N\to+\infty} \P\left(\sup_{0\leq t\leq  NT } X(s)\geq \eps N\right)=0,
\]
consequently
\[
\lim_{N\to+\infty}\P\left(\left\{\sup_{0\leq t\leq T} \left(Y_*^N{-} Z^N\right)(t) \geq \lceil \eps N\rceil\right\}\bigcap {\cal A}_N\right)=1.
\]
We conclude the proof by using the fact that  $Y_*^N(t){>}0$  implies that $Z^N(t){=}0$. 
\end{proof}
The following corollary gives a more precise statement concerning the asymptotic behavior of the $Z$-component of the state vector. It is a simple consequence of the fact that  $Y_*^N(t){>}0$ implies $Z^N(t){=}0$.
\begin{corollary} \label{thm:z-is-null}
Under Condition~\eqref{eq:cond1}, there exists $t_1{>}0$ such that
\[
\lim_{N\to+\infty} \P\left(  Z^N(t)=0, \forall t \in [t_1,T] \right)= 1
\]
holds for any initial fluid state~\eqref{ScalingInit}. 
\end{corollary}

\begin{theorem}[Saturated regime]\label{theoinstab}   Under Condition~\eqref{eq:cond1}, there exists $t_1{>}0$ such that,  for any initial fluid state~\eqref{ScalingInit}, 
any limiting point $(y_{\infty}^*(t), y_{\infty}(t), z_{\infty}(t))$ of  the sequence $(\overline{X}^N(t))$ defined by Relation~\eqref{eqX} satisfies the following relations, for all $t{\geq}t_1$, 
  $z_{\infty}(t){=} 0$ and the differential equations
  \begin{equation*}
  \begin{aligned}
    \frac{\diff {y}_{\infty}^*}{\diff t}(t)&=\mu_{01}y_{\infty}(t)-\mu_{02}r,\\
    \frac{\diff {y}_{\infty}}{\diff t}(t)&=-\mu_{01}y_{\infty}(t)+p(\mu_{02}r+\mu_{11}(1-y_{\infty}^*(t)-y_{\infty}(t))).
  \end{aligned}
  \end{equation*}
hold.
  Furthermore
  \[
  \lim_{t\to+\infty}\left({y}_{\infty}^*(t), {y}_{\infty}(t)\right)= \left(1 {-}  \frac{\mu_{02} r}{\mu_{01}}\left( \frac{(1{-}p)\mu_{01}}{p \mu_{11}} {+} 1 \right), \frac{\mu_{02} r}{\mu_{01}}\right).
  \]
\end{theorem}
\begin{proof}
  By Proposition~\ref{thm:ystarminusz}, for some $\eps_0$ sufficiently small, there exists $t_1$ such that, for any $T{>}t_1$, the event
  \[
{\cal E}_N\steq{def}\left\{ Y_*^N(t){\geq} \eps_0 N: \forall t\in[t_1,T]\right\}
  \]
  has a  probability arbitrarily close to~$1$ as $N$ gets large. Consider  the process $(Y_{a*}^N(t),Y_{a}^N(t))$ defined in Section~\ref{sec:auxiliary-saturation} with initial state $(Y_{*}^N(t_1),Y^N(t_1))$, then by checking $Q$-matrix of both processes, it is easily seen that, on the event ${\cal E}_N$, the relation
  \[
\left(  (Y_{*}^N(t),Y^N(t)),t_1{\leq}t{\leq}T\right)\steq{dist}\left(  (Y_{a*}^N(t),Y_{a}^N(t)),t_1{\leq}t{\leq}T\right)
\]
holds. By using the fact that the sequence of random variables
\[
\left(\frac{1}{N}(Y_{*}^N(t_1),Y^N(t_1))\right)\in[0,1]
\]
is tight, one has only to use Proposition~\ref{thm:blocY} to conclude the proof of the theorem. 
\end{proof}
\subsection{The Underloaded Regime\label{undersec}}
 \addcontentsline{toc}{section}{\hspace{2mm}\thesubsection. The Underloaded Regime}
 In this section, it will be assumed that Condition~\eqref{eq:cond2} holds. 
\begin{proposition} \label{thm:majorize-y-ystar}
Under Condition~\eqref{eq:cond2}, there exists $\eta_0{>}0$ and  $t_1{>}0$ such that,  for any initial fluid state~\eqref{ScalingInit} and for $T{>}t_1$, 
    \[
      \lim_{N \to \infty} \P \left( \sup_{t \in [t_1, T]} \left(Y_*^N(t){+} Y^N(t)\right) \leq N\left(\underline{y}{-} \eta_0\right)\right) = 1,
      \]
      with $\underline{y}{=}r{\mu_{02}}/{\mu_{01}}$.
\end{proposition}
Note that $\underline{y}{<}1$ by Condition~\eqref{eq:cond2}.
\begin{proof}
By using the SDEs~\eqref{eq:integr-ystar-diff} and~\eqref{eq:integr-y-diff}, we get that 
\begin{multline}\label{f1}
  \diff \left(Y_*^N{+}Y^N\right)(t) = \sum_{i=1}^{+\infty}\ind{i\leq N-Y_*^N(t-)-Y^{N}(t-)}\markedN{\mu_{11},i}(\diff t, \{0\}) \\
 -  \ind{Y_*^N(t-)>0}\markedN{\mu_{02}C_2^N}(\diff t, \{1\})
 -  \sum_{i=1}^{+\infty}\ind{i\leq Y^N(t-), Z^N(t-)>0} \markedN{\mu_{01},i}(\diff t, \{1\})
\end{multline}
holds.  The strategy of the proof is of deriving an upper bound for the process $(Y_*^N(t){+}Y^N(t))$, as before we will work on the differential terms of the above relation.

We choose $\eta{>}0$ sufficiently small so that for $N$ large enough the relation
\[
\frac{\mu_{02}}{\mu_{01}} C_2^N > \left\lfloor N\left(r\frac{\mu_{02}}{\mu_{01}} {-}\eta \right)\right\rfloor 
  \]
  holds.  Under this condition one has, with a convenient coupling of Poisson processes,
\begin{equation}\label{f2}
\markedN{\mu_{02}C_2^N}(\diff t, \diff u )\geq \sum_{i=1}^{+\infty}    \ind{i \leq N (\mu_{02} r / \mu_{01} - \eta)} \markedN{\mu_{01},i}(\diff t, \diff u).
\end{equation}
The relation $Y_*^N(t-){>}0$ implies $Z^N(t-){=}0$, consequently, we get the inequality
\begin{equation*}
\ind{Y_*^N(t-)>0}\markedN{\mu_{02}C_2^N}(\diff t, \{1\})
\geq \ind{Z^N(t-)=0} \sum_{i=1}^{+\infty}    \ind{i \leq N (\mu_{02} r / \mu_{01} - \eta)} \markedN{\mu_{01},i}(\diff t, \{1\}).
\end{equation*}
If the relation $Y^N(t-){+} Y_*^N(t-) {\geq}N(\mu_{02} r /\mu_{01}{-} \eta)$ holds, then
\begin{multline}\label{f3}
\ind{Z^N(t-)>0} \sum_{i=1}^{+\infty} \ind{i\leq Y^N(t-)}\markedN{\mu_{01},i}(\diff t, \{1\}) \\
  \geq  \ind{Z^N(t-)>0}\sum_{i=1}^{+\infty}\ind{i \leq N (\mu_{02} r / \mu_{01} - \eta)} \markedN{\mu_{01},i}(\diff t, \{1\}),
\end{multline}
since $Y_*^N(t{-}){=}0$ if $Z^N(t{-}){>}0$.

By plugging Relations~\eqref{f2} and~\eqref{f3} into the SDE~\eqref{f1}, we get that
\begin{equation}\label{f4}
 \diff\left(Y_*^N{+}Y^N\right)(t) \leq {\cal N}_{\lambda_N}(\diff t) -{\cal N}_{\mu_N}(\diff t) 
\end{equation}
holds on the event $Y^N(t-){+} Y_*^N(t-) {\geq}N(\mu_{02} r /\mu_{01}{-} \eta)$, with 
\[
(\lambda_N,\mu_N)\steq{def}\left(\left\lfloor p\mu_{11} N \left(1- r\frac{\mu_{02}}{\mu_{01}} + \eta\right)\right\rfloor,\left\lfloor  (1-p) \mu_{01} N \left(r\frac{\mu_{02}}{\mu_{01}} - \eta\right)\right\rfloor \right). 
\]
By Condition~\eqref{eq:cond2} we can take $\eta{=}\eta_0{>}0$ to be such that
\[
  2\eta_0 < r\dfrac{\mu_{02}}{\mu_{01}} - \dfrac{p \mu_{11}}{p \mu_{11} + (1-p) \mu_{01}},
\]
In this case if  $({\lambda},{\mu})$ is the limit of the sequence $((\lambda_N,\mu_N)/N)$, then ${\lambda}{<}{\mu}$. There exist $0{<}\lambda_0{<}\mu_0$, such that, for $N$ sufficiently large, the relations $\lambda_N{\leq}\lambda_0 N$ and $\mu_N{\geq}\mu_0 N $ hold. Let $(X(t))$ be the (ergodic) $M/M/1$ queue with input [resp.\ service] rate given by $\lambda_0$ [resp. $\mu_0$] and $X(0){=}N$, then Equation~\eqref{f4} gives the coupling relation $(Y_*^N{+}Y^N)(t){\leq}X(Nt)$ for $t$ less than the  hitting time of $N(\mu_{02} r /\mu_{01}{-} 2\eta_0)$.  Consequently, by ergodicity,  there exists some $t_1{\geq}0$ such that this hitting time is, with high probability, less than $t_1N$. Now, by taking  the initial conditions $(Y_*^N{+}Y^N)(0){=}X(0){=}\lceil N(\mu_{02} r /\mu_{01}{-} 2\eta_0)\rceil$, with the same argument as in the proof of Proposition~\ref{thm:ystarminusz}, one gets that, for $T{>}0$,  the process $((Y_*^N{+}Y^N)(t))$ remains below $N(\mu_{02} r /\mu_{01}{-} \eta_0)$ on the time interval $[0,T]$ with high probability. The proposition is proved. 
\end{proof}
The following result is the analogue of Proposition~\ref{thm:ystarminusz} for the underloaded regime. 
\begin{proposition} \label{thm:z-minus-y-star}
Under Condition~\eqref{eq:cond2}, for any $\eps{>}0$ small enough, there exists a time $t_1{\geq}0$ such that,   for any initial fluid state~\eqref{ScalingInit} and for $T{>}t_1$,
\[
        \lim_{N \to \infty} \P \left( \inf_{t \in [t_1,T]} (Z^{N}(t){-}Y_*^N(t)) \geq \eps N \right) = 1.
\]
\end{proposition}
\begin{proof}
Since the proof follows the same lines as in the proof of Proposition~\ref{thm:ystarminusz}, we sketch the main technical arguments. 
  From the last proposition, one can  chose $\eta_0$ and $t_1{\geq}0$ such that the event
\[
{\cal B}_N\steq{def}\left\{ \sup_{t\in [t_1,T]} (Y_*^N(t){+}Y^N(t)) \leq N\left(r\frac{\mu_{02}}{\mu_{01}}- \eta_0\right) \right\}
\]
has a probability  converging  to $1$ when $N$ gets large. 

The SDEs~~\eqref{eq:integr-y-diff} and~\eqref{eq:integr-z-diff} give the relation
\begin{equation*}
  \diff \left(Z^{N}{-}Y_*^N\right)(t) = \sum_{i=1}^{+\infty} \ind{i \leq C_2^N{-}Z^N(t^-)}{\cal N}_{\mu_{02},i}(\diff t)
  - \sum_{i=1}^{+\infty} \ind{i\leq Y^N(t-)}{\cal N}_{\mu_{01},i}(\diff t),
\end{equation*}
by using again that $Z^N(t{-})$ is null if $Y_*^N(t{-})$ is positive.

One takes  $\eta_0{<}\mu_{02}\eps /(4\mu_{01})$ then, 
on the event ${\cal B}_N$, if  $Z^N(t) {\leq} \eps N$, 
\begin{multline*}
  \diff\left(Z^{N}{-}Y_*^N\right)(t) \geq \sum_{i=1}^{+\infty} \ind{i \leq C_2^N - \eps N}{\cal N}_{\mu_{02},i}(\diff t)
\\  - \sum_{i=1}^{+\infty} \ind{i\leq N \mu_{02}(r - \eps/2) / \mu_{01}}{\cal N}_{\mu_{01},i}(\diff t).
\end{multline*}
Hence,  the  process $(Z^N{-}Y_*^N(t))$ can be compared with a (scaled) ergodic  $M/M/1$ queue with arrival rate $\lfloor C_2^N{-}\eps N\rfloor \mu_{02}$ and service rate $\lfloor N(r{-}\eps/2)\mu_{02}/\mu_{01}\rfloor \mu_{01}$.
We conclude the proof in the same way as in the proof  of Proposition~\ref{thm:yplusystar}.
\end{proof}
\begin{corollary} \label{thm:y-star-is-null}
Under Condition~\eqref{eq:cond2}, there exists $t_1{>}0$ such that
\[
\lim_{N\to+\infty} \P\left(  Y_*^N(t)=0, \forall t \in [t_1,T] \right)= 1
\]
holds for any initial fluid state~\eqref{ScalingInit}. 
\end{corollary}

\begin{proof}
  The proof follows from the mutual exclusivity of the events $\{ Y_*^N(t) {>} 0 \}$ and $\{Z^N(t) {>} 0 \}$ and from Proposition~\ref{thm:z-minus-y-star}.
\end{proof}

We can now state the main result for the underloaded regime. 
\begin{theorem}[Underloaded Regime]\label{theostab}
  Under Condition~\eqref{eq:cond1}, there exists $t_1{>}0$ such that,  for any initial fluid state~\eqref{ScalingInit}, 
any limiting point $(y_{\infty}^*(t), y_{\infty}(t), z_{\infty}(t))$ of  the sequence $(\overline{X}^N(t))$ defined by Relation~\eqref{eqX} satisfies the following relations, for all $t{\geq}t_1$, 
$y^*_{\infty}(t){=}0$ and the differential equations 
  \begin{equation*}
  \begin{aligned}
    \dfrac{\diff y_{\infty}}{\diff t}(t) &= - (p \mu_{11} + (1-p) \mu_{01} ) y_{\infty}(t) + p \mu_{11} \\
    \dfrac{\diff z_{\infty}}{\diff t} (t) &=  - \mu_{02} z_{\infty}(t) - \mu_{01} y_{\infty}(t) + \mu_{02} r
  \end{aligned}
  \end{equation*}
  hold. 
Furthermore, 
  \begin{equation*}
\lim_{t\to+\infty}  (y_{\infty}(t),z_{\infty}(t))=\left(\dfrac{p \mu_{11}}{p \mu_{11}{+}(1{-}p) \mu_{01}},  r{-}\dfrac{p\mu_{01}\mu_{11}}{\mu_{02}(p \mu_{11}{+}(1{-}p)\mu_{01})}\right). 
  \end{equation*}
\end{theorem}
\begin{proof}
  In the same way as in the proof of Theorem~\ref{theoinstab}, a coupling between the processes  $(Y_{b}^N(t),Z_{b}^N(t))$ defined in Section~\ref{sec:auxiliary-free} and the process  $(Y_N(t),Z_N(t))$ can be constructed so that the convergence results of Proposition~\ref{prop5} can be used. 
  \end{proof}

\providecommand{\bysame}{\leavevmode\hbox to3em{\hrulefill}\thinspace}
\providecommand{\MR}{\relax\ifhmode\unskip\space\fi MR }
\providecommand{\MRhref}[2]{%
  \href{http://www.ams.org/mathscinet-getitem?mr=#1}{#2}
}
\providecommand{\href}[2]{#2}

\end{document}